\documentclass[12pt,oneside,reqno,dvipsnames]{amsart}
\usepackage{graphicx}% Include figure filse
\usepackage{caption,subcaption}
\usepackage{dcolumn}% Align table columns on decimal point
\usepackage{amsmath,amsthm,amssymb}
\usepackage{todonotes}
\usepackage{url}
\usepackage{placeins}

\usepackage{booktabs}

\usepackage{euscript}

\newtheorem*{lemma*}{Lemma}
\newtheorem*{rem*}{Remark}

\usepackage{amsthm}
\usepackage{latexsym}
\usepackage{amssymb}
\usepackage{amsmath}
\usepackage{mathrsfs}

\begin{document}

\title[AdS Robin solitons and their stability]
{AdS Robin solitons and their stability}

\author{Piotr Bizo\'n}
\address{Institute of Physics, Jagiellonian
University, Krak\'ow, Poland}
\email{bizon@th.if.uj.edu.pl}

\author{Dominika Hunik-Kostyra}
\address{Institute of Physics, Jagiellonian
University, Krak\'ow, Poland}
\email{dominika.hunik@uj.edu.pl}

\author{Maciej Maliborski}
\address{Gravitational Physics, Faculty of Physics, University of
  Vienna, Boltzmann\-gasse 5, A-1090 Vienna, Austria}
\email{maciej.maliborski@univie.ac.at}

% ----------------------------------------------------------------
\begin{abstract}

We consider the four-dimensional  Einstein-Klein-Gordon-AdS system
 with conformal mass  subject to the Robin boundary conditions at infinity. Above a critical value of the Robin parameter, at which the AdS spacetime goes linearly unstable, we prove existence of  a family of globally regular static solutions (that we call AdS Robin solitons) and discuss their properties.

  \end{abstract}
\maketitle

\section{Introduction}
We consider the four-dimensional  Einstein-Klein-Gordon-AdS system
 with mass $\mu$ related to the negative cosmological constant $\Lambda$ through $\mu^2=\frac{2}{3} \Lambda$. For this, and only this, value of mass the system is conformally well-behaved at null and spatial infinity and consequently the initial-boundary value problem is well-posed for a variety of different boundary conditions at infinity \cite{f2,hw}. Here, we focus on the one-parameter family of  Robin boundary conditions. It has been known that along this family there is a critical parameter value at which the system undergoes a bifurcation: the (zero energy) anti-de Sitter (AdS) spacetime becomes linearly unstable  above that critical value \cite{iw} and there emerges a pair of (negative energy) globally regular static solutions (henceforth called AdS Robin solitons) \cite{hh}. The main goal of this paper is to establish the existence of AdS Robin solitons rigorously and analyze the structure of the bifurcation in more detail. In preparation of  future analysis of the role of solitons in dynamics, we also determine their  spectrum of linearized perturbations.

\section{Setup}
The Einstein-Klein-Gordon-AdS  system is given by
\begin{subequations}
\begin{align}\label{ekg}
& G_{\alpha\beta}+ \Lambda g_{\alpha \beta} =
 8 \pi G
  \left(\partial_{\alpha} \phi \,\partial_{\beta} \phi - \frac{1}{2} \left(g^{\mu\nu} \partial_{\mu}\phi \, \partial_{\nu}\phi +\mu^2 \phi^2\right)\, g_{\alpha\beta}\right)\;,\\
&  \Box_g \phi - \mu^2 \phi=0,
\end{align}
\end{subequations}
where $\Box_g=g^{\alpha\beta} \nabla_{\alpha}\nabla_{\beta}$ is the wave operator associated with  the metric $g_{\alpha\beta}$, $\mu$ is the mass of the scalar field, and $\Lambda$ is a negative constant. We assume  spherical symmetry
and write  the metric in the form
\begin{equation}\label{g}
 g= \frac{\ell^2}{\cos^2{\!x}}\left( -A e^{-2 \delta} dt^2 + A^{-1} dx^2 + \sin^2{\!x} \, d\omega^2\right)\,,
\end{equation}
where $(t,x,\omega)\in (-\infty,\infty) \times [0,\pi/2) \times \mathbb{S}^2$, $d\omega^2$ is the round metric on $\mathbb{S}^2$ and $\ell^2=-3/\Lambda$. The metric functions $A,\delta$ and the scalar field $\phi$ depend on $(t,x)$. We choose units such that $\ell=1$ and $4\pi G=1$ and  introduce  new variables
\begin{equation}
f= \frac{\phi}{\cos{x}}\quad\mbox{and} \quad B = \frac{A-1}{\cos^2{\!x}}\,.
\end{equation}
Then the system (1) reduces to
\begin{subequations}
\begin{align}
(\Box_{\hat{g}}-1) f =& \frac{2+\mu^2}{\cos^2 x} f-\left( 1-3\cos^2 x\right) B f-\mu^2\sin^2 x \; f^3\;,
\\
\cos x \; \partial_x B =& -\frac{B}{\sin x} - \sin x \; (1+B \cos^2 x)\;\Phi-\mu^2 \sin x \; f^2\;,
\\
\partial_x \delta =& -\sin x\cos x \; \Phi\;,\\
\partial_t B=& -2 A \sin{x} \;(\cos x \,\partial_x f -f \sin x)\; \partial_t f\;,
\end{align}
\end{subequations}
where $\Phi = (\cos x \,\partial_x f -f \sin x)^2+A^{-2} e^{2\delta} \cos^2 x \; (\partial_t f)^2$ and
$$
\Box_{\hat{g}}=-e^{\delta} \partial_t\left(A^{-1} e^{\delta} \partial_t\right)+\frac{e^{\delta}}{\sin^2{\!x}}\,\partial_x\left(A e^{-\delta} \sin^2{\!x} \, \partial_x\right)
$$
is the polar wave operator associated with the conformal metric $\hat g_{\alpha\beta}=\cos^2{\!x}\, g_{\alpha\beta}$.
On the right side of equation (4a) the derivatives of metric functions were eliminated using equations (4b) and (4c).

  In the following we set $\mu^2=-2$. For this value of mass the wave equation (4a) is regular at $x=\pi/2$ because the first  term on the right  side (which is the only singular term) vanishes\footnote{This cancellation is due to the fact that for $\mu^2=-2$ the left sides of equations (1b) and (4a) are asymptotically conformal, that is
 $(\Box_g+2)\phi \approx (\cos{x})^{-3} (\Box_{\hat g}-1) f$ near $x=\pi/2$. However, the constraint equation (4b) has a singularity at $x=\pi/2$ so it does not appear possible to extend the solutions `beyond infinity' (cf. \cite{f3} where an extension of solutions across the conformal boundary at \emph{timelike} infinity was analyzed for the system (1) with $\mu^2=\frac{2}{3} \Lambda>0$).}.
   Thanks to this fact, the initial-boundary value problem for the system (4) is well posed for a variety of boundary conditions  at the conformal boundary (both reflective and dissipative) \cite{f2,hw}\footnote{This should be contrasted with the widely studied massless case  for which only the Dirichlet boundary condition is compatible with the basic requirement of finite  total mass  \cite{br}.}. In this paper we focus our attention on the one-parameter family of Robin boundary conditions
\begin{equation}\label{robin}
\partial_x f-b\, f\vert_{x=\frac{\pi}{2}}=0,
\end{equation}
where $b$ is a constant (hereafter referred to as the Robin parameter). For $b=0$ the Robin condition reduces to the Neumann condition $\partial_x f\vert_{x=\pi/2}=0$.

Assuming \eqref{robin} and expanding the fields in power series in $z=\pi/2-x$
we obtain the following asymptotic behavior near $z=0$
\begin{align}
f(t,x) &=\alpha - b \alpha z + \mathcal{O}(z^2)\;,\\
B(t,x) &= \alpha^2 - (3 b \alpha^2 +M) z + \mathcal{O}(z^2)\;,\\
\delta(t,x) &=\delta_{\infty}+\frac{1}{2}\alpha^2 z^2 + \mathcal{O}(z^3)\;,
\end{align}
where $\alpha(t)$ and $\delta_{\infty}(t)$ are free functions\footnote{We use the normalization $\delta(t,0)=0$, hence $t$ is the proper time at the center.} and $M$ is a constant. To see the physical meaning of $M$, let us define the renormalized mass function
\begin{equation}\label{m}
m=-B \tan{x} + \frac{\sin^3{x}}{\cos{x}}\, f^2\,.
\end{equation}
The first term on the right side is the Misner-Sharp  mass function defined by $m_{MS}=r (1+r^2-g^{\mu\nu} \partial_{\mu} r \partial_{\nu} r)$, where $r=\tan{x}$ is the areal radial coordinate. This function diverges as $x\rightarrow \pi/2$ and the purpose of the second term (called the counterterm) is to cancel this divergence. The leading order behavior of the counterterm is determined by the asymptotics (6) and (7) but otherwise  can be chosen freely.
Using equation (4b) we get
\begin{equation}\label{mprim}
  \partial_x m=\rho \sin^2{x}\;,
\end{equation}
where
\begin{equation}\label{rho}
\rho=A^{-1} e^{2\delta} (\partial_t f)^2 +(\partial_x f)^2 + f^2 +B (\cos{x}\; \partial_x f -f \sin{x})^2\;,
\end{equation}
hence
\begin{equation}\label{m-vol}
m(t,\pi/2)=\int_{0}^{\pi/2} \rho \sin^2{x}\,dx.
\end{equation}
This  quantity can be interpreted as the bulk energy.
From the asymptotic expansions  (6) and (7) it follows that
\begin{equation}\label{masym}
M=m(t,\pi/2)-b \alpha^2(t),
\end{equation}
where the second term on the right side can be viewed as the energy stored on the boundary. Although both the bulk and  boundary energies are time dependent, their sum $M$ is conserved. In what follows, we will refer to $M$ as the total energy (mass). The exchange of energy between the bulk and the boundary is a characteristic feature of systems subject to the  Robin boundary conditions. Note that some of the bulk  energy is ``lost" to the boundary if $b<0$ and  ``gained" from the boundary if $b>0$.

We remark that the expression \eqref{masym} can be obtained in a systematic way  within the diffeomorphism  covariant Hamiltonian framework of Wald and Zoupas \cite{wz} (see section 2.2 in \cite{hhol}).
Nonetheless, we believe that our hands-on approach, based solely on the analysis of the system (4), is helpful in getting insight into not so widely known physics of the Robin boundary conditions.
\section{AdS Robin solitons}
For  time-independent solutions  the system (4) with $\mu^2=-2$ takes the form
\begin{subequations}
  \begin{align}
    \label{eq:1}
&(1+B\cos^2{\!x}) f''+\cot{x} \left(2+(1-4\sin^2{\!x}) B + 2 \sin^2{\!x}\, f^2\right) f'-f \nonumber\\
 &\qquad\qquad\qquad\,\,\, +\left(1-3\cos^2{\!x}\right) B f-2\sin^2{\!x} \; f^3=0\;,
    \\
    \label{eq:2}
    & \cot{x} \; B' +\frac{B}{\sin^2{\!x}} + (1+B \cos^2{\!x})\;(\cos x \,f'-\sin{x} \,f)^2 -2 f^2=0\;,\\
    \label{eq:3}
& \delta' + \sin{x} \cos{x}\,(\cos x \,f'-\sin{x} \,f)^2=0\,,
\end{align}
\end{subequations}
where the derivatives of  metric functions were eliminated from equation (\ref{eq:1}) using equations (\ref{eq:2}) and (\ref{eq:3}). It is routine to prove that this system has local solutions near $x=0$ which behave as follows
\begin{equation}\label{static0}
f(x) \sim c +\frac{1}{6} c x^2,\quad B(x) \sim \frac{2}{3} c^2 x^2,\quad \delta(x)\sim -\frac{1}{9} c^2 x^4\,,
\end{equation}
where $c$ is a free parameter.
\begin{lemma*}
For any $c$ the local solution \eqref{static0} extends smoothly up to $x=\pi/2$ and fulfills the boundary conditions (6)-(8).
\end{lemma*}
\begin{proof}
 To prove this lemma it is convenient to use the radial coordinate $r=\tan{x}$ and return to the original field variables
\begin{equation}\label{}
  \phi(r)= f(x) \cos{x},\qquad A(r)=1+ B(x) \cos^2{x}.
\end{equation}
Then, equations \eqref{eq:1} and \eqref{eq:2} become
\begin{subequations}
  \begin{align}
    \label{eqphi}
&  (1+r^2) A \phi''+\left(r(1+r^2)A\phi'^2+(1+r^2)A'+\frac{2+4 r^2}{r} A\right)\phi'+2\phi=0, \\
\label{eqA}
&   (1+r^2)A'-\frac{1+3 r^2}{r} (1-A) -2 r \phi^2
+ r(1+r^2) A \phi'^2 = 0.
 \end{align}
\end{subequations}
The local solutions \eqref{static0} translate to
\begin{equation}\label{ics}
\phi(r)\sim c-\frac{1}{3} c r^2,\quad A(r)\sim 1+\frac{2}{3} c^2 r^2.
\end{equation}
We first observe that the function $B=(A-1)(1+r^2)$ is monotone increasing.
To see this, suppose that $B(r)$ has a maximum at some point $r_0>0$. Differentiating equation \eqref{eqA}, substituting $B'(r_0)=0$, and eliminating $\phi''(r_0)$ and $B(r_0)$ using equations \eqref{eqphi} and \eqref{eqA}, respectively, we get  after simplifications
\begin{equation}\label{b2prim}
  B''(r_0) = 2 \phi'^2+4 r^2 \phi^2 \phi'^2+4 (r \phi'+\phi)^2\vert_{r=r_0}\,,
\end{equation}
which is manifestly positive, contradicting that the point $r_0$ exists. Since $B(r)$ is positive for small $r>0$, this implies that $A(r)\geq 1$ for all $r$.
\vskip 0.1cm
\noindent Next, we define a function
\begin{equation}\label{H}
H=\frac{1}{2} (1+r^2) A \phi'^2 + \phi^2.
\end{equation}
Using the system (17) we obtain
\begin{equation}\label{dH}
  H'=-\frac{(1+r^2) A (3+r^2 \phi'^2)+2 r^2 \phi^2 + 3 r^2+1}{2r}\,,
\end{equation}
which is manifestly negative, hence $H(r)$ is a monotonically decreasing Lyapunov function. Since $A\geq 1$, it follows that $r^2 \phi'^2$ and $\phi^2$ remain bounded for all $r$.
\vskip 0.1cm
To determine the asymptotic behavior of solutions for $r\rightarrow \infty$ it is convenient to use the logarithmic radial variable $\tau=\log{r}$. In terms of $\tau$ the system (17)
is asymptotically autonomous for $\tau\rightarrow \infty$ and the limiting autonomous system is
\begin{subequations}
\begin{align}\label{eqphi_lim}
 & A \ddot \phi + (3+2\phi^2) \dot\phi +2 \phi=0,\\
\label{eqA_lim}
  &\dot A-3(1-A)-2\phi^2+A\dot\phi^2=0,
\end{align}
\end{subequations}
where dot denotes the derivative with respect to $\tau$ (by an abuse of notation, we use the same symbols for the original and limiting systems). From the general theory of asymptotically autonomous dynamical system \cite{aa, th} and the existence of the Lyapunov function $H$, it follows that the asymptotic behavior of solutions of the system (17) for $\tau\rightarrow \infty$ is governed by the above limiting system. Elementary analysis gives the attracting fixed point $\phi=0,\dot\phi=0,A=1$ with the leading order behavior
\begin{equation}\label{focus}
  \phi(\tau) = c_1 e^{-\tau} + c_2 e^{-2\tau}+\mathcal{O}(e^{-3\tau}),\qquad A(\tau)-1  = c_1^2 e^{-2\tau} + c_3 e^{-3\tau}+\mathcal{O}(e^{-4\tau}),
\end{equation}
where $c_k$  are free parameters, which are related to the free parameters $\alpha, b$, and $M$ in the expansions (6) and (7) by
\begin{equation}\label{ck}
  c_1=\alpha,\quad c_2=-b \alpha,\quad c_3=-3b \alpha^2 -M.
\end{equation}
This completes the proof.
\end{proof}

 The above Lemma ensures that for each $c$ the solution starting with the initial conditions \eqref{static0}
 automatically satisfies the Robin condition $f'(\pi/2)=b f(\pi/2)$ for some parameter $b$ (which depends on $c$). We will refer to these globally regular static solutions as the AdS Robin solitons (or just solitons for short) and denote them by $(f_s,B_s,\delta_s)$.
 The profiles of solitons can be easily determined numerically by integrating the system (14) with the boundary conditions \eqref{static0}. We note in passing that an analogous reasoning leads to a two-parameter family of hairy black holes  (where the second parameter is the horizon radius).

As far as we know,  the AdS Robin solitons and hairy black holes were first studied in the literature in the context of so called ``designer gravity" \cite{hh, hm}, however, to the best of our knowledge,  their existence remained unproven.

\section{Bifurcation analysis}
It is illuminating to look at the solitons from the viewpoint of the local bifurcation theory. To this end, consider the perturbation expansion of solitons for small $c$
\begin{equation}\label{pert}
 b=b_*+c^2 b_2+\mathcal{O}(c^4), \quad  f=c f_1 + c^3 f_3+\mathcal{O}(c^5),\quad B=c^2 B_2+ c^4 B_4+\mathcal{O}(c^6)\,.
 \end{equation}
 Inserting this expansion into the system (14) and requiring regularity at $x=0$,
at the lowest order we get
\begin{equation}\label{order1}
b_*=\frac{2}{\pi},\quad f_1=\frac{x}{\sin{x}}, \quad B_2=\frac{x}{\sin{x}}\,\left(-\cos{x}+\frac{x}{\sin{x}}\right)\,.
\end{equation}
At the third order equation (\ref{eq:1}) becomes
\begin{align}
  \label{eq:4}
f_3''&+2\cot x\;f_3'-f_3=-B_2\cos^2 x\;f_1''-\cot x \;((1-4\sin^2 x)B_2\nonumber \\
&+2\sin^2 xf_1^2)f_1'-(1-3\cos^2 x)B_2f_1+2\sin^2 x f_1^3\;.
\end{align}
Substituting $f_1$ and $B_2$, given in \eqref{order1}, into the right hand side and imposing  $f_3(0)=0$, we find
\begin{align}\label{f3a}
f_3(x) = \frac{1}{12\sin x} & \left( \left( 5-24 \zeta\left(3 \right) \right) x
 + 3x\cos{(2x)}\;-\frac{2x^3}{\sin^2 x}-3\sin{(2x)}\;\right. \nonumber \\
& \left. + 16 x^3 C_1(x) - 48 x^2 S_2(x) - 72x C_3(x) + 48 S_4(x)  \right),
\end{align}
where $\zeta$ is the Riemann zeta function and we defined the functions
\begin{equation}
S_n(x) = \sum_{k=1}^{\infty} \frac{\sin(2kx)}{k^n} , \qquad C_n(x) = \sum_{k=1}^{\infty} \frac{\cos(2kx)}{k^n} \;.
\end{equation}
From \eqref{f3a} we read off
\begin{align}
& f_3\left(\frac{\pi}{2}\right) = \frac{\pi}{48}\left(4-\pi^2(8\ln 2+1)+60\zeta(3)\right)\approx 0.754316\;,\label{f3pi2}\\
& f'_3\left(\frac{\pi}{2}\right) = \frac{2}{3}+\frac{\pi^2}{8}\left(8\ln 2 -1\right)- \frac{7}{2}\zeta\left(3\right) \approx 2.06686\;.\label{df3pi2}
\end{align}
At the fourth order equation (14b) becomes
\begin{align}
\cot x\; B_4'&+\frac{1}{\sin^2 x}B_4 = -B_2 \cos^2 x\;(\cos x\;f_1'-\sin x f_1)^2\nonumber \\
&-2(\cos x\;f_1'-\sin x f_1)(\cos x\;f_3'-\sin x f_3)+4f_1f_3\;.
\end{align}
Substituting \eqref{order1} and \eqref{f3a}  into the right hand side and requiring regularity at $x=0$, we find
\begin{align}\label{b4a}
B_4(x) &= \frac{1}{6}x \left(\cot x-\frac{x}{\sin^2 x} \right) \left(24\zeta(3)-5 \right) +4x \left( 2\cot x-\frac{3x}{\sin^2 x} \right)C_3(x)\nonumber \\
&+4x^2 \left(\cot x-\frac{2x}{\sin^2 x} \right)S_2(x)
 -\left( 6\cot x-\frac{8x}{\sin^2 x} \right)S_4(x)+\frac{3}{4} \cos^2{x}  \nonumber \\ & +\left(\frac{3}{4}+\frac{1}{3} x^2\right)x^2 \cot x
+ \frac{8x^4}{3\sin^2 x} C_1(x) +x^2 \left(\frac{1}{4\sin^2 x}-\frac{3}{4}-\frac{1}{3} x^4\right) \nonumber \\ &+\frac{2x^2}{\sin^2 x} - x \cot x \left(1-x^2+\frac{1}{4}\cos (2x)+\frac{x^2}{2\sin^2{x}}\right)\,,
\end{align}
from which we read off
\begin{equation}\label{b4p}
B'_4\left(\frac{\pi}{2}\right) = \frac{\pi}{48}\left(54+\pi^2(32\ln 2-11)\right)\approx 10.7566\;.
\end{equation}
Using  \eqref{f3pi2} and \eqref{df3pi2} and
 imposing  the Robin condition in the expansion \eqref{pert}, we get
\begin{equation}\label{b2}
  b_2=\frac{2}{\pi} \,\left(f_3'-\frac{2}{\pi} f_3\right)\big\vert_{\pi/2}=\frac{\pi}{6} (16 \ln{2} -1)+ \frac{1}{\pi}(1-12\zeta(3)) \approx 1.01009\,.
\end{equation}
The fact that $b_2$ is positive means that at $b_*$ we have a  supercritical pitchfork bifurcation where the AdS solution  bifurcates into a pair of solitons ($\pm f_s, B_s, \delta_s$). As usual, this kind of bifurcation is associated with exchange of linear stability and, indeed, in the next section we will show that for $b>b_*$ the AdS space becomes linearly unstable  whereas  the solitons  are linearly stable.

Using (7) and the expansion  \eqref{pert},  we get the approximation for the mass
\begin{equation*}
M_s= B'(\pi/2)-3 b \alpha^2 \simeq \frac{3\pi}{2} c^2+B_4'(\pi/2) c^4 - 3 \left(\frac{2}{\pi}+b_2 c^2\right) \left(\frac{\pi}{2} c+ f_3(\pi/2) c^3\right)^2,
\end{equation*}
which upon substitution of \eqref{f3pi2} and \eqref{b4p} yields
\begin{equation}\label{mass-soliton}
  M_s \simeq -\frac{\pi^2 b_2}{8} \, c^4=-\frac{\pi^2}{8 b_2}\, (b-b_*)^2\,.
\end{equation}

 It is instructive to rederive this result  along the lines of designer gravity \cite{hh}. Letting $\alpha=f(\pi/2)$ and $\beta=f'(\pi/2)$, we get from \eqref{pert} (in this paragraph `$=$'  means equality up to order $\mathcal{O}(c^4)$)
\begin{equation}\label{alpha-beta}
\alpha= f_1(\pi/2)\, c +  f_3(\pi/2)\,c^3,\quad \beta= f'_1(\pi/2)\,c+ f'_3(\pi/2)\,c^3\,,
\end{equation}
which can be viewed as the parametric equation of the curve in the $(\alpha,\beta)$ plane. Eliminating $c$ we get the function
\begin{equation}\label{beta}
\beta_s(\alpha)=\frac{2}{\pi} \alpha +\frac{4 b_2}{\pi^2}  \alpha^3,
\end{equation}
where the subscript `s' indicates that the function is associated with solitons.
Following the approach used in  designer gravity  we introduce the effective potential
\begin{equation}\label{potential}
  \mathcal{V}(\alpha)=2\int_0^{\alpha} \beta_s(\alpha') d\alpha' -  b \alpha^2=-(b-b_*) \alpha^2 +\frac{2b_2}{\pi^2} \alpha^4\,.
\end{equation}
By construction, critical points of the effective potential correspond to solitons. The key observation, made by Hertog and Horowitz in \cite{hh}, is that the value of the effective potential at the critical point is equal to the soliton mass. In our case, $\mathcal{V}'(\alpha)=0$ for
$\alpha_s^2=\frac{\pi^2}{4 b_2} (b-b_*)$ (and, of course, for $\alpha=0$ corresponding to the AdS space). Substituting this into \eqref{potential} we get  $M_s=\mathcal{V}(\alpha_s)$ which reproduces the formula \eqref{mass-soliton}. Note that $\mathcal{V}''(\alpha_s)>0$.

  Further from the bifurcation point, the soliton function $\beta_s(\alpha)$ and the corresponding effective potential $\mathcal{V}(\alpha)$ can be determined numerically\footnote{For large values of $c$, there develops a boundary layer near $x=\pi/2$  with exponentially shrinking width. Using the method of matched asymptotics one can show that both $\alpha$ and $b$ grow as $e^{c^2}$ for $c\rightarrow \infty$ which makes the numerics (in compactified variable $x$) cumbersome.}. We find that for each $b>b_*$ the effective potential has the shape of a Mexican hat (see Fig.~2) with exactly three critical points:  the local maximum at zero and two  global minima at $\pm \alpha_s$. This implies that the soliton solution is unique (modulo reflection symmetry) and suggests that it is stable.
    \begin{figure}[h!]
    \centering
\includegraphics[width=0.6\textwidth]{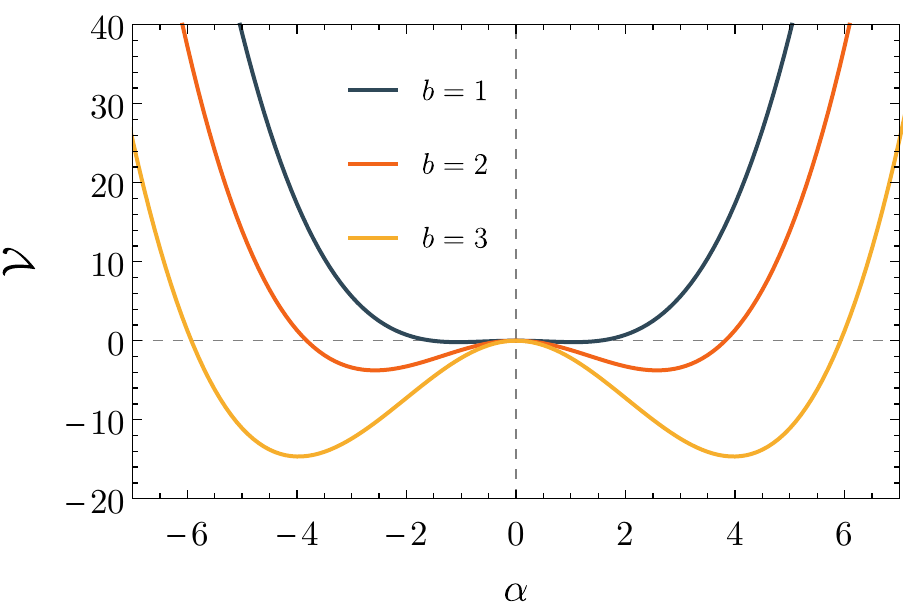}
\caption{\label{fig1} {\small Effective potentials for sample values of $b$.}}
\end{figure}
\vskip 0.2cm
\noindent
\emph{Remark.} It is natural to expect that for any given $b>b_*$ the soliton is the ground state, i.e.
   for any regular initial data satisfying the Robin condition \eqref{robin} there holds the inequality $M\geq M_s$ which saturates if and only if the data correspond  to the soliton   \cite{hh}.
   Our numerical constructions of initial data corroborate this conjecture but we have not been able to prove it (see \cite{hhol} for partial results in this direction).

\section{Linear stability analysis}
Linearizing the system (4) around the AdS solution ($f=B=\delta=0$) and separating time $f(t,x)=e^{i\omega t} v(x)$ we get the eigenvalue problem\footnote{The eigenvalue problem \eqref{eigen}  is a particularly simple  case of the master eigenvalue problem for linear perturbations of AdS space that was solved  by Ishibashi and Wald in full generality using the properties of hypergeometric functions \cite{iw}. For the reader's convenience we  reproduce their results in our special case using more elementary tools.}
\begin{equation}\label{eigen}
L v = \omega^2 v, \quad \mbox{where} \,\,\,L=-\frac{1}{\sin^2{\!x}}\,\partial_x\left(\sin^2{\!x}\, \partial_x\right)+1.
\end{equation}
The operator $L$ (which is just the polar conformal Laplacian on the 3-sphere) is symmetric on the Hilbert space $L^2\left([0,\pi/2], \sin^2{\!x}\, dx\right)$
and the Robin boundary condition
\begin{equation}\label{rc}
v'-b v\vert_{x=\frac{\pi}{2}}=0
\end{equation}
 provides a one-parameter family of its self-adjoint extensions.

For $\omega^2>0$ the  regular solution of \eqref{eigen}  is
\begin{equation}\label{v}
  v(x)=\frac{\sin(\omega x)}{\sin{x}}.
\end{equation}
Imposing the Robin condition \eqref{rc} we obtain the quantization condition
for the eigenfrequencies
\begin{equation}\label{robin-omega}
\omega=b \tan\left(\omega\pi/2\right).
\end{equation}

\begin{figure}[h]
  \includegraphics[height=0.38\textwidth]{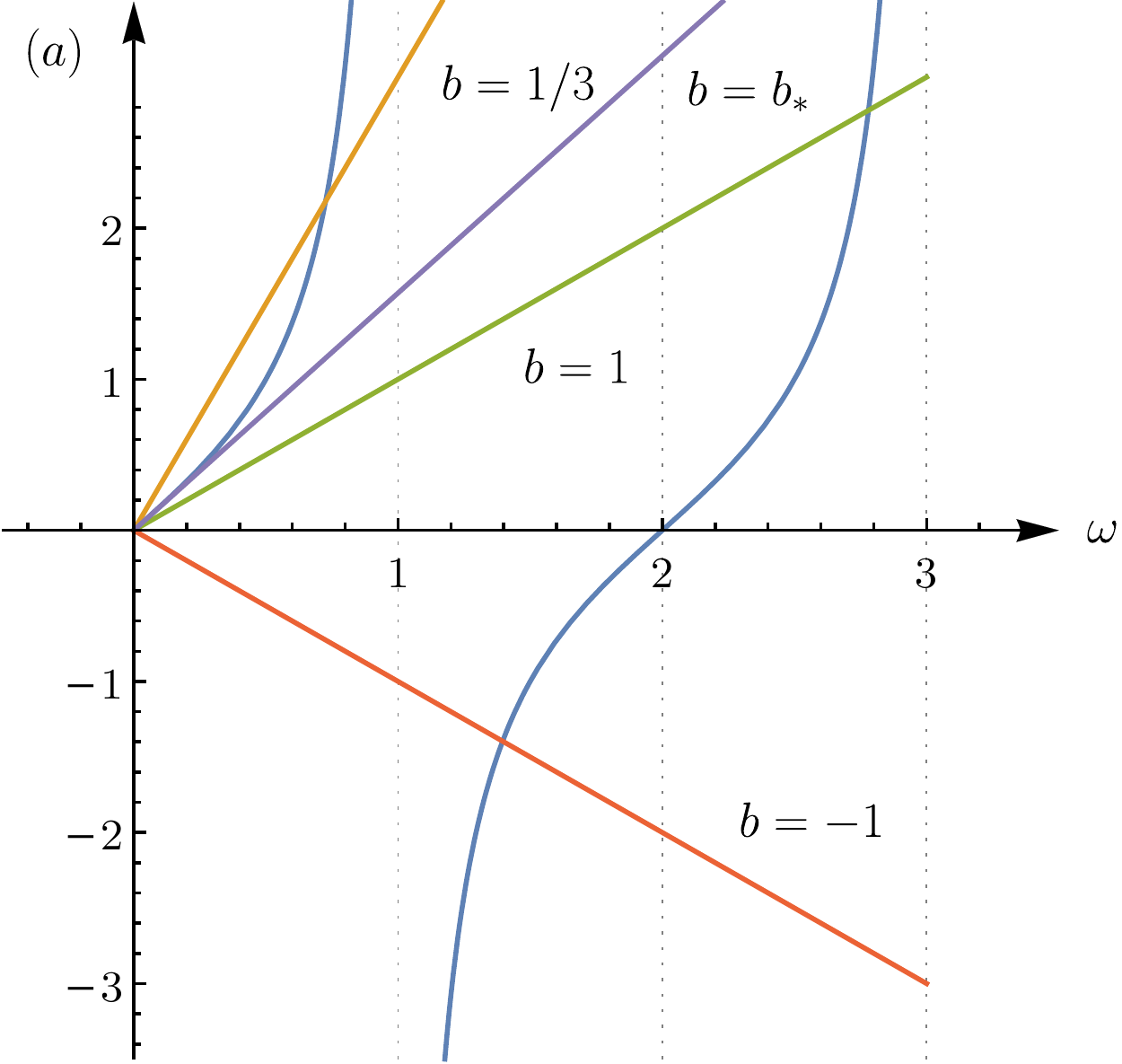}
  \hspace{4ex}
  \includegraphics[height=0.38\textwidth]{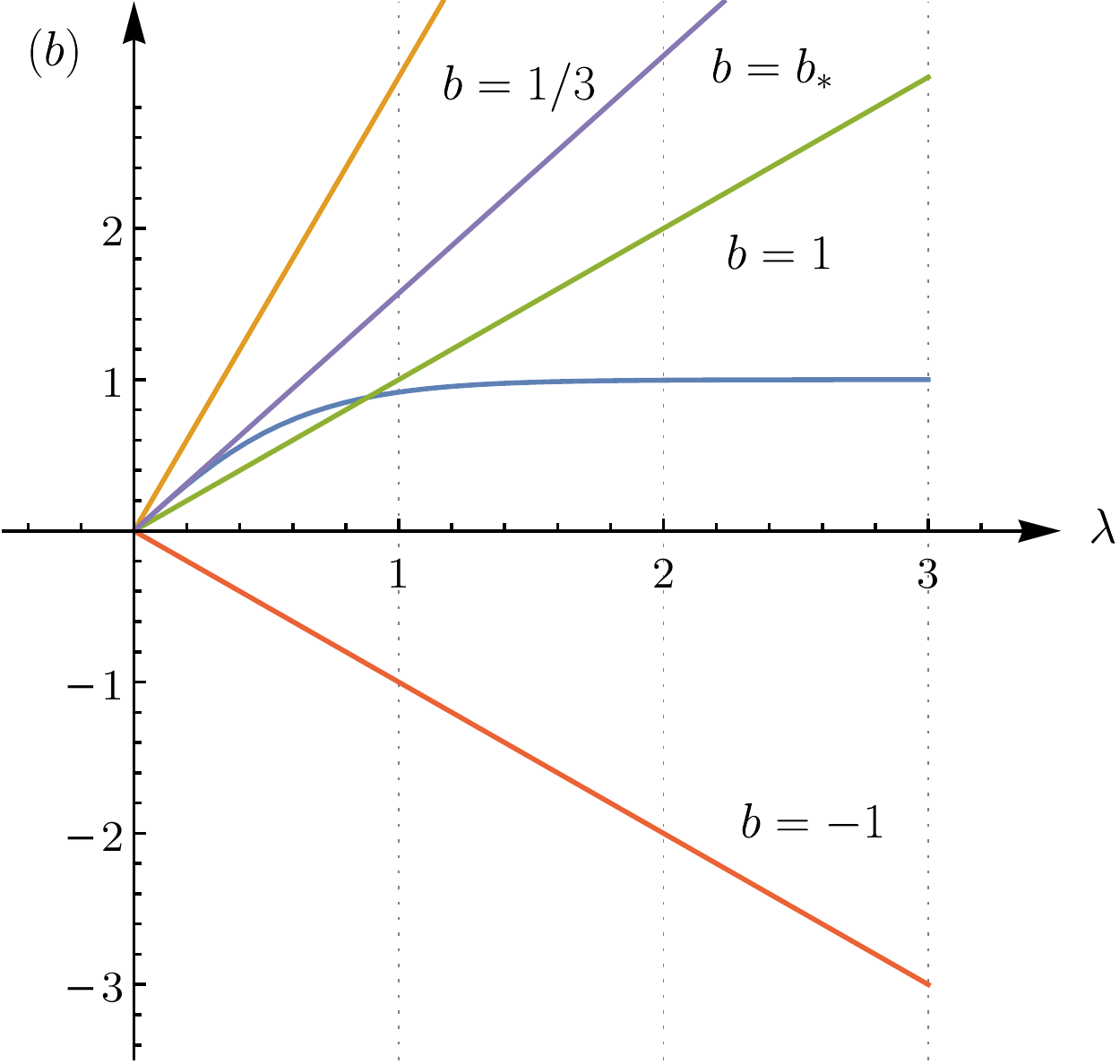}
  \captionsetup{width=\textwidth}
  \label{fig5}
  \caption{{\small Graphical solutions of the quantization
        conditions \eqref{robin-omega} and \eqref{robin-lambda}.}}
\end{figure}
From the graphical analysis shown in Fig.~2a we see that for each non-negative integer $n$ there is exactly one eigenfrequency $\omega_n$ such that
\begin{align*}
2n+1 &< \omega_n<2n+2 \quad \mbox{if} \quad b<0,\\
2n &<\omega_n<2n+1 \quad \mbox{if} \quad 0<b<\frac{2}{\pi}\;.
\end{align*}
For large $n$ the quantization condition \eqref{robin-omega} gives the asymptotically resonant spectrum
\begin{equation}\label{eigen-asym}
\omega_n=2n+1-\frac{b}{\pi n} +\mathcal{O}\left(\frac{1}{n^2}\right)\,.
\end{equation}

  The lowest eigenvalue $\omega_0^2$ vanishes at $b=b_*=2/\pi$; the corresponding  eigenfunction is the linearized static solution $f_1$ given in \eqref{order1}.  An elementary  perturbative calculation gives near $b_*$
  \begin{equation}
  \omega_0^2 \approx \frac{6}{\pi} (b_*-b).
  \end{equation}
  For general $b>b_*$ there is an exponentially growing mode $e^{\lambda_0 t} v_0(x)$, where  the exponent $\lambda_0=\sqrt{-\omega_0^2}$  is given by the unique positive root of the equation (see Fig.~2b)
 \begin{equation}\label{robin-lambda}
 \lambda=b \tanh\left(\lambda\pi/2\right)
  \end{equation}
 and the corresponding eigenfunction is $v_0(x)=\sinh(\lambda_0 x)/\sin{x}$.
\vskip 0.2cm
 Next, we look at the linear stability of solitons. Linearizing the system (4) around the soliton and separating time, we get the eigenvalue problem
\begin{equation}\label{eigen-sol}
L_s v = \tilde \omega^2 v,
\end{equation}
where
\begin{equation}\label{Ls}
L_s=-\frac{A_s e^{-\delta_s}}{\sin^2{\!x}}\,\partial_x\left(A_s e^{-\delta_s} \sin^2{\!x}\, \partial_x\right)+ A_s e^{-2\delta_s} U
\end{equation}
and
\begin{multline}
  U = 1+(3\cos^2{\!x}-1) B_s + 4  \sin^2{\!x}\, (2-\sin^2{\!x})  f_s^2
  \\
  + \sin{x} \cos{x} \,(8 \sin^2{\!x}-4) f_s f_s' - \cos^2{\!x} \,(2+4 \sin^2{\!x})  f_s'^2
  \\
 + 8 \sin^3{\!x} \cos^3{\!x} f_s^3 f_s'- 4 \sin^4{\!x} \cos^2{\!x} f_s^4
 - 4 \sin^2{\!x} \cos^4{\!x} f_s^2 f_s'^2\,.
\end{multline}
For $b$ slightly above $b_*$ (i.e. for small $c$), we have
\begin{equation}\label{P}
  L_s=L+c^2 P +\mathcal{O}(c^4),
\end{equation}
where the operator $P$ can be calculated using the expansions \eqref{pert}. To calculate the perturbations of eigenvalues we assume the following ansatz
\begin{equation}\label{pert-ansatz}
  v_n=c v_n^* + c^3 u_n + \mathcal{O}(c^5),\quad \tilde\omega_n^2={\omega_n^*}^2+\gamma_n c^2+\mathcal{O}(c^4),\quad b=b_*+b_2 c^2+\mathcal{O}(c^4),
\end{equation}
where ${\omega_n^*}^2$ and $v_n^*$ are the eigenvalues and normalized eigenfunctions of the operator $L$ at the bifurcation point and $b_2$ is given in \eqref{b2}. Substituting this ansatz into  the Robin boundary condition we get at the first and third order in $c$
\begin{equation}\label{robin-pert}
 {v_n^*}'(\pi/2)=b_* v_n^*(\pi/2),\qquad  u_n'(\pi/2)=b_* u_n(\pi/2) + b_2 v_n^*(\pi/2)\,.
\end{equation}
 Substituting the ansatz \eqref{pert-ansatz} into \eqref{eigen-sol}, we get at the  third order in $c$
\begin{equation}\label{c3}
  L u_n + P v_n^*={\omega_n^*}^2 u_n+\gamma_n v_n^*\,.
\end{equation}
Projecting this equation on $v_n^*$ and noting, via \eqref{robin-pert}, that
\begin{equation}\label{parts}
  (v_n^*, L u_n)=({v_n^*}' u_n-v_n^* u_n')\vert_{x=\frac{\pi}{2}}+(u_n,L v_n^*)=-b_2 {v_n^*}^2(\pi/2)+{\omega_n^*}^2 (u_n,v_n^*),
\end{equation}
we obtain  the leading order  approximation for the eigenvalues
\begin{equation}\label{born}
  \tilde \omega_n^2(c)\approx {\omega_n^*}^2 + \gamma_n c^2, \qquad \gamma_n=(v_n^*,Pv_n^*)-b_2 {v_n^*}^2(\pi/2).
  \end{equation}
In particular, for the lowest eigenvalue we obtain
\begin{equation}\label{omega0}
  \tilde\omega_0^2 \approx \gamma_0 c^2,
\end{equation}
where
\begin{equation}\label{gamma0}
  \gamma_0=32 \log{2} -2 +\frac{1}{\pi^2} \left(12-144 \zeta(3)\right)\approx 3.8582.
\end{equation}
The positivity of $\gamma_0$  confirms the expectation that the solitons are linearly
stable near the bifurcation point.
Solving the eigenvalue problem numerically, we find  that  the eigenvalues $\tilde \omega_n^2$ grow monotonically with $c$.
The  numerical values of the first few  eigenfrequencies (as measured by the central observer) for a small parameter $c=0.1$ (corresponding to $b\approx0.6467$) are displayed in Table~1.

\begin{table}[h]
  \centering
  \begin{tabular}{|c|cccccc|}
    \toprule
    $n$ & 0 & 1 & 2 & 3 & 4 & 5\\
    \midrule
    $\tilde \omega_n$ & 0.19735 & 2.87065 & 4.93028 & 6.95714 & 8.97363 & 10.98549 \\
    $\tilde \omega_n^{pert}$ & 0.19642 & 2.87062 & 4.93025 & 6.95711 & 8.97360 & 10.98546 \\
    \bottomrule
  \end{tabular}
  \vskip 1ex
  \caption{\small{The first six eigenfrequencies of linear
      perturbations around the soliton for the  parameter
      $c=0.1$. In the second row the approximate eigenfrequencies given by \eqref{born} are shown for comparison.}}
  \label{tab:Eigenfrequencies}
\end{table}
\noindent From the leading order WKB approximation \cite{bo} it follows that for large $n$
\begin{equation}\label{large_n}
  \tilde \omega_n = \frac{2n+1}{a} + \mathcal{O}\left(\frac{1}{n}\right),\qquad a=\frac{2}{\pi}\int_0^{\pi/2} A_s^{-1} e^{\delta_s} dx,
\end{equation}
which compares well with numerical results even if $n$ is not very large.

\section{Discussion}
The Einstein-Klein-Gordon-AdS system with  mass $\mu^2=\frac{2}{3} \Lambda<0$ is  well-behaved at the conformal boundary which makes it a good toy model for studying the role of boundary conditions in dynamics of asymptotically AdS spacetimes~\cite{f5}. In this paper we focused on the Robin boundary conditions and proved existence of a one-parameter family of solitons for $b>b_*$. We also demonstrated that the linearized perturbations around these solitons have no growing modes. A natural question is: are the AdS Robin solitons nonlinearly stable? Numerical simulations, to be reported in \cite{mm}, indicate a positive answer and provide evidence for  existence of plethora of time-periodic and quasiperiodic solutions, not only in the perturbative regime (which is expected in view of the non-resonant spectrum) but also, somewhat surprisingly, for large perturbations.

Of course, the analogous question of nonlinear stability arises for the AdS spacetime for $b<b_*$. However here, in contrast to the Dirichlet case \cite{br}, the numerical simulations are as yet not conclusive and we leave this question to future investigations.
\subsection*{Acknowledgements.} We thank Piotr Chru\'sciel, Oleg Evnin, Helmut Friedrich and Arthur Wasserman for helpful remarks. This work was supported in part by the Polish National Science
  Centre grant no.\ 2017/26/A/ST2/00530. PB and MM
acknowledge the support of the Alexander von Humboldt Foundation.

\end{document}